\documentclass[sigconf,edbt]{acmart-edbt2020}

\def\BibTeX{{\rm B\kern-.05em{\sc i\kern-.025em b}\kern-.08em
    T\kern-.1667em\lower.7ex\hbox{E}\kern-.125emX}}

\usepackage{booktabs} 

\usepackage{times}
\usepackage{soul}
\usepackage{url}
\usepackage[utf8]{inputenc}
\usepackage{graphicx}
\usepackage{amsmath}
\usepackage{booktabs}
\usepackage{algorithm}

\usepackage[show]{notes-alt}
\usepackage[noend]{algpseudocode}

\usepackage{enumitem}

\usepackage{tikz}
\usepackage{pgfplots}
\usepackage{textcomp}
\usepackage[font=Large]{subfig}
\usepackage{multirow}
\usepackage[normalem]{ulem}

\newcommand{\nnp}{\ensuremath{\textsc{SPA}}}
\newcommand{\kse}{\ensuremath{\textsc{SP}}}
\newcommand{\mpp}{\ensuremath{\textsc{SPM}}}
\newcommand{\opp}{\ensuremath{\textsc{SPO}}}
\newcommand{\sbh}{\ensuremath{\textsc{SBP}_\textsc{H}}}
\newcommand{\sbe}{\ensuremath{\textsc{SBP}}}
\newcommand{\nne}{\ensuremath{\textsc{NNE}}}
\newcommand{\ope}{\ensuremath{\textsc{DPE}}}

\newcommand{\leastcomp}{\ensuremath{\textsc{LCMC}}}
\newcommand{\leastcompd}{\ensuremath{\textsc{LCMD}}}
\newcommand{\random}{\ensuremath{\textsc{Random}}}

\newcommand{\TFSN}{\textsc{TFSN}}
\newcommand{\TFSNC}{\textsc{TFSNC}}

\newcommand{\sld}{\textit{Slashdot}}
\newcommand{\wiki}{\textit{Wikipedia}}
\newcommand{\epin}{\textit{Epinions}}

\definecolor{darkgreen}{rgb}{0.125,0.5,0.169}

\def\bp{\textsc{{sbp}}}
\def\ne{\textsc{{nne}}}
\def\pe{\textsc{{dpe}}}
\def\spa{\textsc{{spa}}}
\def\spm{\textsc{{spm}}}
\def\spo{\textsc{{spo}}}

\newtheorem{claim}{Claim}

\setcopyright{rightsretained}

\acmDOI{}

\acmISBN{978-3-89318-083-7}

\acmConference[EDBT 2020]{23rd International Conference on Extending Database Technology (EDBT)}{March 30-April 2, 2020}{Copenhagen, Denmark} 
\acmYear{2020}

\begin{document}
\title{Forming Compatible Teams in Signed Networks}

\author{Ioannis Kouvatis{\small $~^{1}$}, Konstantinos Semertzidis{\small $~^{2}$} \and Maria Zerva{\small $~^{1}$}, Evaggelia Pitoura{\small $~^{1}$}, Tsaparas Panayiotis{\small $~^{1}$}
\vspace{0.1cm}
}

\affiliation{%
  \institution{{\small $~^{1}$} Department of Computer Science and Engineering, University of Ioannina, Greece}
}
\email{{ikouvatis,mzerva,pitoura,tsap}@cse.uoi.gr}
\affiliation{%
  \institution{\vspace{0.1cm}{\small $~^{2}$} IBM Research, Dublin, Ireland}
}
\email{konstantinos.semertzidis1@ibm.com}

\renewcommand{\shortauthors}{}

\begin{abstract}
The problem of team formation in a social network asks for a set of individuals who not only have the required skills to perform a task but who can also communicate effectively with each other.
Existing work assumes that all links in a social network are positive, that is, they indicate friendship or collaboration between individuals. 
However, it is often the case that the network is \emph{signed}, that is, it contains both positive and negative links, corresponding to friend and foe relationships.
Building on the concept of structural balance, we provide definitions of compatibility between pairs of users in a signed network, and algorithms for computing it.
We then define the team formation problem in signed networks, where we ask for a \emph{compatible} team of individuals that can perform a task with small communication cost.
We show that the problem is NP-hard even  when there are no communication cost constraints, and we provide heuristic algorithms for solving it. We present experimental results to investigate the properties of the different compatibility definitions, and the effectiveness of our algorithms.
\end{abstract}

\maketitle

\vspace*{-0.1in}
\section{Introduction}
Given a task that requires a set of skills,  a pool of workers who possess some of the skills and are organized in a social network, team formation refers to finding a subset of the workers
that collectively cover the skills and can communicate effectively with each other \cite{Lappas09}. 
The communication cost of the team is measured using the distances between the team members in the network.
The idea is that socially well connected users will be more effective in working together.

 Since the pioneering work in~\cite{Lappas09}, there has been considerable research activity on the problem~\cite{Kargar11,5590958,Anagnostopoulos:2012:OTF:2187836.2187950}. 
All existing work 
assumes that the social network contains only positive ties between individuals. That is, the edges denote friendship or successful collaboration between two users.
However, very often, we have \emph{signed} networks with both positive and negative ties. A negative edge indicates a contentious relationship and inability of two users to collaborate and thus they should not be in the same team.

In this paper, we study the problem of team formation in signed networks. In addition to the known requirements of the  team formation problem, we ask that the team contains users that are all \emph{compatible} with each other.  Defining compatibility in a signed network is an interesting problem in itself. Clearly, users connected with a positive edge are compatible, while users connected with a negative edge are incompatible. We infer the compatibility of non-connected pairs of users by using the structure of the graph, and the principle of \emph{structural balance}. 
Structural balance \cite{davis1963structural} is based on the premise that
``the friend of my friend is my friend'', ``the enemy of my enemy is my friend'', and ``the enemy of my friend is my enemy''. Using this premise, we determine the compatibility of two users by looking at the paths that connect them. For example,  a path of two positive or two negative edges indicates compatibility, while  a  path of one positive and one negative edge indicates incompatibility. We formalize this idea, and  we provide definitions of compatibility of varying strictness.
We perform experiments with four real datasets, where we evaluate the different compatibility definitions, and the performance of the algorithms for the team formation problem.

\vspace*{-0.05in}
\section{Problem Definition}
\label{sec:problem}

We are given as input a pool of $n$ individuals organized in an \emph{undirected signed} graph $G=(V, E)$. Each node in $V$ corresponds to an individual, and $E = \{(u,v,\ell): u,v\in V, \ell \in\{+1,-1\}\}$ is a set of edges \emph{labeled} as either positive or negative to indicate that $u$ and $v$ are friends or enemies respectively. We assume that $G$ is connected.
We will use a function $sign: E \rightarrow \{+1,-1\}$ that returns the label of each edge in $E$.

We are also given as input a universe $S$ of skills.
Each individual $u$ in $V$ possesses a set of skills, $skill(u) \subseteq S$.
We define a task as the subset of skills $T$ $\subseteq$ $S$  required for its completion.
Given a task, the team formation problem asks for a team $X$ of individuals, $X \subseteq V$, that collectively covers the required skills and whose members can
work together \emph{effectively}.
%
The effectiveness of a team is typically quantified by the communication cost, $Cost(X)$, of the team defined as some function of the distances of its members in the graph. 

However, when the graph contains both positive and negative edges, we need to take into account that some individuals, although close in the graph, may not be \emph{compatible} with each other.
%
%
To capture whether two users are compatible, we introduce a relation  $Comp$ $\subseteq$ $V \times V$
such that $(u, v)$ $\in$ $Comp$, if and only if, $u$ and $v$ can work together.
Two natural requirements for $Comp$ are (1) reflexivity: $(u, u)$ $\in$ $Comp$, and (2) symmetry: if $(u, v)$ $\in$ $Comp$, then $(v, u)$ $\in$ $Comp$.
Furthermore, we require that the $Comp$ relation satisfies the following two intuitive properties:
\begin{enumerate}
\item \emph{Positive Edge Compatibility:} For all $(u,v) \in E$, such that $sign(u,v) = +1$, $(u,v) \in Comp$. 
\item
\emph{Negative Edge Incompatibility:} For all $(u,v) \in E$, such that $sign(u,v) = -1$, $(u,v) \not \in Comp$. 
\end{enumerate}

We will provide various definitions of compatibility in Section~\ref{sec:compatibility}. We now define formally our problem.
\begin{definition} \textsc{(Team Formation in Signed Networks (TFSN))} \,
	Given a signed graph $G=(V, E)$, a compatibility relation $Comp$, and a task $T$, find $X$  $\subseteq$ $V$ such that
	(1)  $\bigcup_{u \in X}skill(u)$ $\supseteq$ $T$,
	(2) for each $u,v\in X$, $(u, v)$ $\in$ $Comp$,
	and (3)  $Cost(X)$ is minimized.
\end{definition}
\vspace*{-0.03in}
The {\TFSN} problem contains as a special case the original team formation problem which is NP-hard \cite{Lappas09}, thus {\TFSN} is also NP-hard.
Moreover, we have shown that just finding a set of compatible users is NP-hard. Let {\TFSNC} denote the simplified version of the {\TFSN} problem, where we drop the third requirement of minimizing the cost. In particular, we have proven the following theorem.
\vspace*{-0.05in}
\begin{theorem}
The decision version of {\TFSNC} is NP-hard for any compatibility relation that satisfies positive edge compatibility and negative edge incompatibility.
\end{theorem}
\section{User Compatibility}
\label{sec:compatibility}
We start with two basic definitions of compatibility.

\vspace*{-0.05in}
\begin{definition} Direct Positive Edge ($\ope$) compatibility:  
$Comp_{\pe} =\{(u,v)\subseteq V\times V: (u,v,+1) \in E\}$.
\end{definition}

\vspace*{-0.1in}
\begin{definition} No Negative Edge ($\nne$) compatibility:
$Comp_{\ne} =\{(u,v) \subseteq V \times V: (u,v,-1) \not\in E\}$.
\end{definition}

\vspace*{-0.05in}
$\ope$  is the strictest form of compatibility, while  $\nne$ is the most relaxed one. Specifically, $Comp_{\pe}$ is the minimal subset of pairs of nodes that satisfies the positive edge compatibility property, while $Comp_{\ne}$ is the maximal subset of pairs of nodes that satisfies the negative edge incompatibility property.

We will now use the theory of \emph{structural balance}
\cite{cartwright1956structural,davis1963structural,book1}
to provide more refined definitions of compatibility.  
The theory is based on the following 
socially and psychologically founded premises:  (1)  the friend of my friend is my friend, (2) the friend of my enemy is my enemy, and (3) the enemy of my enemy is my friend.
Let $P = (v_0, $\dots$, v_{k+1})$,  $(v_{i},v_{i+1}) \in E$ denote a path between nodes $v_0$ and $v_{k+1}$ in a signed graph $G$. 
We define the sign of the path as $sign(P) = \prod_{i = 0...k} sign(v_i,v_{i+1})$. We say that path $P$ is positive if $sign(P) = +1$ and negative if $sign(P) = -1$.

\begin{claim}
A positive path $P_{uv}$ between two nodes $u$ and $v$ indicates compatibility, while a negative one  indicates incompatibility.
\end{claim}
The claim follows from the basic principle of structural balance. To see this,
let $P_{uv} = (x_0,x_1,...,x_k,x_{k+1})$, $x_0 = u$, $x_{k+1} = v$ be a path that connects $u$ and $v$. Let $F_u$ be the set of  friends of $u$ and $E_u$ be the set of  enemies of $u$.
We start by placing node $u$ in $F_u$ and
traverse the path as follows. When we traverse edge $(x_i,x_{i+1})$, if the edge is positive we place $x_{i+1}$ in the same set as $x_i$.
That is, the friends of my friends are also my friends, and the friends of my enemies are my enemies.
If the edge $(x_i,x_{i+1})$ is negative then we place $x_{i+1}$ in the opposite set of $x_i$.
That is, the enemies of my enemies are my friends, and the enemies of my friends are my enemies.
If the path $P_{uv}$ is positive then $v$ will be placed in $F_u$, while if the path is negative it will be placed in $E_u$.

We first look at shortest paths.
We use $SP_{uv}$ to denote the set of shortest paths between nodes $u,v$, $SP_{uv}^+$ to denote the positive, and  $SP_{uv}^-$ the negative ones.


\vspace*{-0.05in}
\begin{definition}
Shortest Path (\kse) compatibility relations:
-- All Shortest Path ({\nnp}) compatibility: $Comp_{\spa} = \{(u,v) \subseteq V\times V: \forall P_{uv}\in SP_{uv}, sign(P_{uv}) = +1\}$. \\
-- Majority Shortest Path ({\mpp}) compatibility: $Comp_{\spm} = \{(u,v) \subseteq V\times V: |SP_{uv}^+| \geq |SP_{uv}^-|\}$. \\
--  One Shortest Path ({\opp}) compatibility: $Comp_{\spo} = \{(u,v) \subseteq V\times V: \exists P_{uv}\in SP_{uv}, sign(P_{uv}) = +1\}$.
\end{definition}

\vspace*{-0.02in}
We further relax compatibility by asking for positive paths that are not necessarily the shortest ones.
Based on structural balance, certain triangles are more stable.
A general signed graph is structurally balanced, 
if it does not contain any cycle with an odd number of negative edges~\cite{book1}.

Given a path $P$, let $G_P = (P, E[P])$ be the graph induced by the nodes of $P$.
We say that path $P$ is structurally balanced if the subgraph $G_P$ is structurally balanced.
Let $BP_{uv}$ denote the set of all structurally balanced paths between $u$ and $v$.

\vspace*{-0.05in}
\begin{definition} Structurally Balanced Path ({\sbe}) compatibility: 
$Comp_{\bp} = \{(u,v) \subseteq V\times V: \exists P_{uv}\in BP_{uv}, sign(P_{uv}) = +1\}$.
\end{definition}
\vspace*{-0.05in}

The motivation  for {\sbe} compatibility is that,
in addition to $P_{uv}$ being positive, asking for $G_P$ to be structurally balanced means that
the sign of any edge connecting $u$ and $v$ must be positive, otherwise a cycle with an odd number of negative edges will be created.
Note that  {\sbe}-compatibility does not imply {\kse}-compatibility. Consider the example in Figure~\ref{fig:counterexample}(a). The (only) shortest path between $u$ and $v$ is $(u,x_1,v)$ which is negative, and thus $u,v$ are not {\kse}-compatible. However, $u$ and $v$ are {\sbe}-compatible, since the path $(u,x_2,x_3,x_4,v)$ is positive and structurally balanced.
Note that there is a shorter path $(u,x_2,x_1,v)$ between $u$ and $v$ that is positive, but  not structurally balanced, since the shortcut edge $(u,x_1)$ creates the unbalanced triangle $(u,x_1,x_2)$.

%

It is easy to see that the following holds:
\begin{proposition}
$Comp_{\pe}$	$\subseteq$  $Comp_{\spa}$ $\subseteq$  $Comp_{\spm}$ $\subseteq$ $Comp_{\spo}$ $\subseteq$ $Comp_{\bp}$ $\subseteq$  $Comp_{\ne}$.
\end{proposition}\vspace{-0.1in}

\begin{algorithm}[t]
	\caption{The {\kse}-compatibility  algorithm.}
	\label{algo:spc}
	\begin{algorithmic}[1]
	 \small 
		\Statex{{\bf Input:} Signed graph $G$, query node $q$.}
		\Statex{{\bf Output:} The number of positive and negative shortest paths from $q$ to all other nodes in the graph.}
		\vspace{0.1cm}
		\hrule
		\vspace{0.1cm}
		\State Initialize $N^+(q)$ = 1, $N^-(q)$ = 0 $N^+(x)$ = $N^-(x)$ = 0, $L(q) = 0$, $L(x) = \infty$, empty queue $Q$.
		\State $Q$.enqueue($q$)
		\While {$Q \neq 0$}
			\State $u$ = $Q$.dequeue()
			\For{$x$ adjacent to $u$}
				\If {$L(u)+1 \leq L(x)$} 
                \label{line:if}
					\If {$x \notin Q$} \State Q.enqueue($x$)
					\EndIf
					\State $L(x) = L(u) + 1$
					\If {$sign(u,x) = +1$}
						\State $N^+(x)$ += $N^+(u)$;  $N^-(x)$ += $N^-(u)$
					\ElsIf {$sign(u,x) = -1$}
						\State $N^-(x)$ += $N^+(u)$; $N^+(x)$ += $N^-(u)$
					\EndIf
				\EndIf
			\EndFor
		\EndWhile
		\State \textbf{return} $(N^+,~N^-,~L)$			
	\end{algorithmic}
\end{algorithm}

\paragraph{Algorithms.}
We  now present algorithms for {\kse} and {\sbe} compatibility.
%
Algorithm~\ref{algo:spc} shows the modified BFS algorithm for counting positive and negative shortest paths.
Given the query node $q$, for each node $x \in V$ in the graph, the algorithm maintains the numbers $N^+(x)$ and $N^-(x)$ of positive and negative shortest paths respectively and the length of the shortest path $L(x)$ from $q$ to $x$. When reaching node $x$ from node $u$ through a shortest path (line~\ref{line:if}), if the edge $(u,x)$ is positive, we increment the number of positive and negative paths of $x$ by $N^+(u)$ and $N^-(u)$ respectively, since all paths retain their sign. If the edge $(u,x)$ is negative, we increment $N^-(x)$ by $N^+(u)$, and $N^+(x)$ by $N^-(u)$, since the sign of the paths change. Each edge is examined only once. 


The efficient enumeration of shortest paths is possible due to the  prefix property that 
 a shortest path between $q$ and $x$ that goes through node $u$ must use a shortest path from $q$ to $u$. However, this is not the case for shortest structurally balanced paths.
 Consider the example in Figure~\ref{fig:counterexample}(b). The shortest structurally balanced path from $u$ to $x_4$ is $(u,x_3,x_4)$. However, the shortest structurally balanced path  $(u,x_1,x_2,x_4,x_5,v)$ from $u$ to $v$ goes through node $x_4$ but not through the shortest structurally balanced path from $u$ to $x_4$, since the path $(u,x_3,x_4,x_5,v)$ is not structurally balanced.

\begin{figure}
	\centering
		\vspace*{-0.1in}
		\resizebox{0.45\textwidth}{!}{
	\subfloat[][]{\includegraphics[clip, trim=2.7cm 11cm 23.5cm 3.6cm]{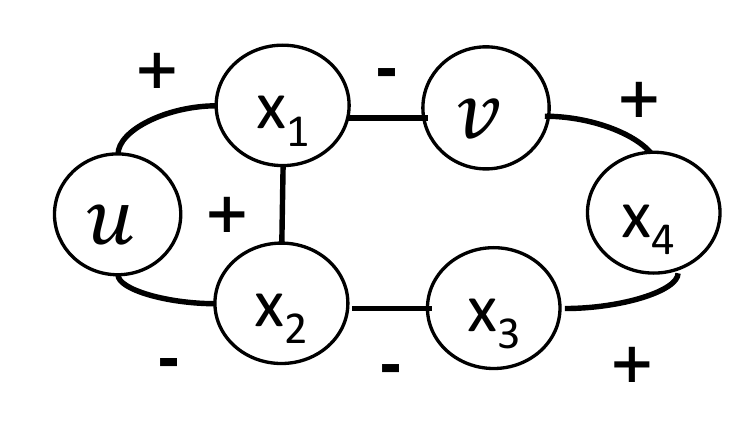}}
		\subfloat[][]{\includegraphics[clip, trim=1.9cm 11cm 19.1cm 3.6cm]{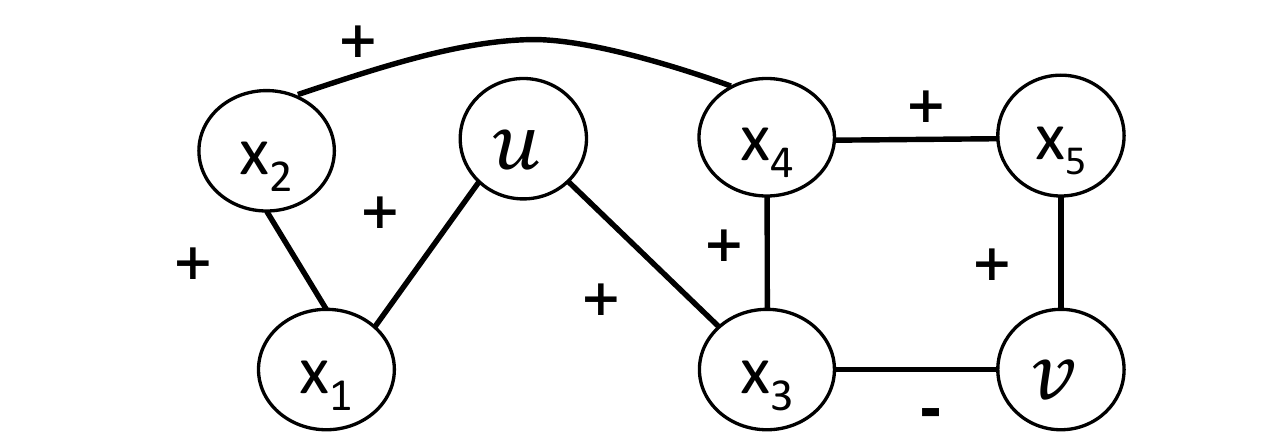}}}
	\vspace*{-0.1in}
	\caption{(a) $u$ and $v$ are {\sbe} but not {\kse} compatible. (b) It does not suffice to keep  a single path from $u$ to $x_4$.
	}	\vspace*{-0.1in}
	\label{fig:counterexample}
\end{figure}

Since the exact algorithm is prohibitively expensive for large graphs, due to the exponential number of paths,  we also consider a heuristic alternative for {\sbe}-compatibility that counts only paths having the prefix property. 
We will use {\sbe} to denote the compatibility relation computed by the exact exhaustive algorithm, and {\sbh} the output of the heuristic algorithm.

\section{Team Formation}
\label{sec:teamform}

We now present algorithms for the {\TFSN} problem. Recall that our goal is to find a team of compatible users, that covers all skills, and minimizes the communication cost. The communication cost is defined as the largest distance between any two pairs of users in the team. We define the distance between two users looking at the positive paths connecting them.
Specifically, for {\ope} and {\kse} compatibility, distance is the length of the shortest path, while  for {\sbe} the length of the  shortest structurally balanced positive path. For {\nne} compatibility, since there be no positive paths, we define distance as the length of the shortest path ignoring its sign.

Algorithm~\ref{algo:tf} is a generic algorithm
that incrementally builds a solution, each time considering an uncovered skill and adding a compatible user having this skill, until all skills are covered.
There are two placeholders in this algorithm. 
The first is the policy for selecting a skill (lines \ref{line:ss1} and \ref{line:ss2}), and the second the policy for selecting a candidate user (line \ref{line:su}).


\begin{algorithm}[t]
	\caption{Team formation algorithm.}
	\label{algo:tf}
	\begin{algorithmic}[1]
	 \small
		\Statex{{\bf Input:} Signed graph $G$, task $T$, compatibility relation $Comp$.}
		\Statex{{\bf Output:} Team $X$.}
		\vspace{0.04cm}
		\hrule
		\vspace{0.1cm}
		\State Initialize $S \leftarrow \emptyset$.
		 \textit{\,\,\,\,//$S$: skills covered so far}
		 	\State Initialize $\mathcal{L} \leftarrow \emptyset$.
		 \textit{\,\,\,\,//$\mathcal{L}$: set of candidate teams}
		\State Select skill $s$ from $T$ \textit{\,\,\,\,//skill selection} \label{line:ss1}
		\For{each $u$ with skill  $s$}
		\State $X$ $\leftarrow$  $\{u\}$  \textit{\,\,\,\,// $X$: candidate team}
		\State $S$ $\leftarrow$ $S$ $\cup$ $(T \cap skills(u))$
		\While{$S$ $\neq$ $T$}
			\State Select skill $s$ from $T - S$ \textit{\,\,\,\,//skill selection} \label{line:ss2}
			\State Select user $v$ with skill $s$ \textit{\,\,\,\,//user selection} \label{line:su}
			\State \,\,\,\,\,\, s.t. $(v,x)\in Comp$ for all $x$ in $X$
			\State $X$ $\leftarrow$ $X$ $\cup$ $\{v\}$
			\State $S$ $\leftarrow$ $S$ $\cup$ $(T \cap skills(v))$
		\EndWhile
		\State $\mathcal{L}$ $\leftarrow$ $\mathcal{L}$ $\cup$  $X$
		\EndFor
		\State \textbf{return} $argmin_{X \in \mathcal{L}} (Cost(X))$				
	\end{algorithmic}
\end{algorithm}

We consider two policies for selecting skills: select the rarest skill first (as in~\cite{Lappas09}), and select the least compatible skill first.
We define the compatibility degree $cd(s)$ of  skill $s$ based on
the compatibility between the users with
skill $s$ and the users with all other skills:
$cd(s)$ = $\sum_{s_j \in S, s_i\neq s}cd(s, s_j)$, where 
$cd(s, s_j)$ = $|\{(u_i, u_j)$ $:$ $(u_i, u_j)$ $\in$ $Comp$, $s \in skills(u_i)$ and $s_j \in skills(u_j)\}|$.

\begin{table}[!ht]
	\centering
	\vspace*{-0.05in}
	\caption{\label{tab:dataset-statistics}Dataset Statistics}
	\vspace*{-0.09in}
	\resizebox{0.77\columnwidth}{!}{
		\begin{tabular}{lcccc}
			\hline
			& {\sld} & {\epin} & {\wiki} \\
			\hline 
			\#users & 214 & 28,854  & 7,066 \\
			\#edges & 304 & 208,778 & 100,790\\
			\#neg edges & 89 (29.2\%) & 34,941  (16.7\%)  & 21,765  (21.5\%)\\
			diameter & 9 & 11 & 7\\
			\#skills & 1,024 & 523 & 500 \\ \hline
	\end{tabular}}
\end{table}
\begin{table}[ht!]
	\centering
	\vspace*{-0.1in}
	\caption{Comparison of compatibility relations}
	\vspace*{-0.1in}
	\label{table:percentage_comp}
	\resizebox{0.8\columnwidth}{!}{
		\begin{tabular}{ccccccc}\hline
			& {\nnp} & {\mpp} & {\opp} & {\sbh} & {\sbe} & {\nne} \\\hline
			\multicolumn{1}{l}{} & \multicolumn{6}{c}{{\sld}}     \\\hline
			comp. users & 44.72 & 55.72 & 72.45 & 97.85 & 99.38 & 99.64\\
			comp. skills & 80.57 & 86.19 & 92.63 & 99.11 & 99.47 & 99.50\\
			avg distance & 4.13 & 4.37 & 4.57 & 4.95 & 4.97 & 4.53\\\hline
			\multicolumn{1}{l}{} & \multicolumn{6}{c}{{\epin}}     \\\hline
			comp. users & 29.61 & 62.98 & 86.46 & 99.82 & -- & 99.99 \\
			comp. skills & 97.25 & 98.90 & 99.66 & 99.99 & -- & 99.99\\
			avg distance & 3.48 & 3.82 & 3.87 & 3.97 & -- & 3.83\\\hline
			\multicolumn{1}{l}{} & \multicolumn{6}{c}{{\wiki}}    \\\hline
			comp. users & 21.98 & 59.33 & 87.51 & 99.56 & -- & 99.91\\
			comp. skills & 66.17 & 87.31 & 97.32 & 99.87 & -- & 99.96  \\
			avg distance & 2.85 & 3.23 & 3.30 & 3.38 & -- & 3.25\\\hline
	\end{tabular}}
\vspace*{-0.1in}
\end{table}

\begin{table}[ht!]
	\centering
	\caption{Comparison with unsigned team formation}
	\vspace*{-0.1in}
	\label{table:baseline}
	\resizebox{0.73\columnwidth}{!}{
		\begin{tabular}{c c c c c  c}\hline
			& {\nnp} & {\mpp} & {\opp}  & {\sbe} & {\nne} \\ \hline
			Ignore sign & 0\% & 2\% & 2\% & 26\% & 30\% \\
			Delete negative & 0\% & 2\% & 18\% & 66\% & 76\% \\ \hline
	\end{tabular}}
	\vspace*{-0.15in}
\end{table}

We also consider two policies for selecting users: select the user that has the minimum distance, and select the user that is most compatible among the remaining users.
The first selection aims at minimizing the cost, while the second at maximizing the chances of finding a group of compatible users.
We experimentally evaluate different combinations of these policies in Section~\ref{sec:experiments}.


\section{Experimental Evaluation}
\label{sec:experiments}

In this section, we compare the different compatibility relations on real datasets and evaluate the team formation algorithm.
%

\noindent{\textbf{Datasets.} Table \ref{tab:dataset-statistics} details our real-world datasets. 
{\sld} contains information about users and their posts on Slashdot. 
We obtained a network of Slashdot users  \cite{social-status}, where users have tagged their relationships as friend or foe. Then we used the categories of users' posts as skills.
{\epin} contains information about users and their reviews about products.
The dataset is created by combining a signed network of Epinions users \cite{social-status} with 
the RED  \footnote{https://projet.liris.cnrs.fr/red/} dataset  which contains information about the products and product categories the users have reviewed.
We used the unique user ids to match users in the two datasets, and we assigned as skills to users the categories of the products they have reviewed.
{\wiki} \cite{social-status} is a signed network of editors. The  edge sign corresponds to a positive or negative vote in admin elections.
Since there was no skill information, we assigned synthetically generated skills to its users.
We generated 500 distinct skills with frequencies following a Zipf distribution as in real data. Each skill is assigned to users in the network uniformly at random.



\noindent\textbf{Compatibility Relations}.
In Table~\ref{table:percentage_comp}, we report the percentage of compatible pairs of users  
and skills. 
The {\ope} is excluded from our analysis, since this corresponds to finding cliques and team formation is too restrictive.
Two skills $s_1$ and $s_2$ are compatible if they have compatibility degree $cd(s_1, s_2) > 0$, i.e., there is at least one compatible pair of users $(u,v)$ such that $u$ has $s_1$ and $v$ has
$s_2$ (including self-compatibility, if the same user has both skills).

As expected, the number of compatible user and skill pairs increases as we relax the notion of compatibility.
For {\nnp}, less than half of the pairs of nodes are compatible, and as low as 21.98\% for the case of {\wiki}. Also, in most cases, for a sizeable fraction of pairs of skills there are no compatible users, indicating that for many skill combinations there can be no compatible team.
Another interesting observation is that the fraction of compatible pairs for {\sbe}  is comparable with that for {\nne}. 
This means that, for all pairs that are not directly connected with a negative edge, there exists at least one positive structurally balanced path that connects them.

\noindent\textbf{\textit{Distance.}} In Table \ref{table:percentage_comp}, we also report the average distance between compatible users.
The distance steadily increases as we relax the compatibility definition.
The exception is {\nne}
in which we allow negative paths, and thus we are able to discover shorter ones.

\noindent\textbf{\textit{Comparison of {\sbe} and {\sbh}.}}
We also compare  the exact (\sbe) and heuristic (\sbh) structurally balanced compatibility
for the {\sld} dataset, for which we can compute the exact relation.
Table \ref{table:percentage_comp} shows the difference between {\sbe} and {\sbh} which is only {\raise.17ex\hbox{$\scriptstyle\sim$}}2.5\%

\noindent\textbf{Team Formation.}
Due to space limitation, in this set of experiments, we only report results using {\epin}. Results are similar for the other  networks. We generate tasks of different sizes. 
For a given task of size $k$, we generated 50 tasks by randomly selecting $k$ skills.
First, we compare the four different team formation algorithms obtained by combining the two different skill and user selection policies.
We report results for the two algorithms that performed the best which
 are the algorithms that select the least compatible skill. 
 The {\leastcompd}, selects  the user with the minimum distance, while the {\leastcomp}, the user who is the most compatible with the existing team.
We also experiment with a baseline {\random} that selects a compatible user at random.

In Figure \ref{fig:sol_perc}(a), we report the percentage of times that
each algorithm was able to  find  a compatible team for $k$ = 5.
The last bar ({\sc MAX}) shows the percentage of tasks that contain compatible skills.
This is a rough upper bound on the number of compatible teams, since it is
based on compatible skills and not the compatibility of users. 
The two algorithms perform equally well indicating that optimizing for compatibility makes very little difference.
Figure \ref{fig:sol_perc}(b) shows the average cost of the teams produced and
indicates that {\leastcompd} is the best choice.
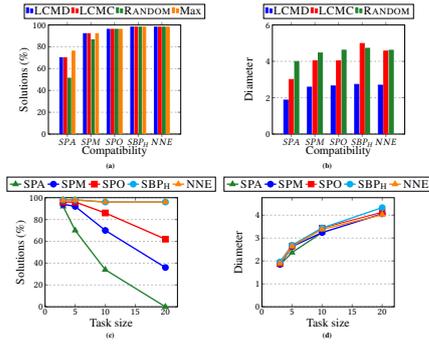
\begin{figure}[t!]
\centering
\vspace{-0.1in}
\resizebox{0.7\columnwidth}{!}{
\begin{tabular}{cc}
\subfloat[]{
\begin{tikzpicture}
\begin{axis}[
	ybar,
	enlargelimits=0,
	legend style={legend columns=4, font=\Huge, anchor=north,at={(0.5,1.19)}},
	xlabel={Compatibility},
    ylabel={Solutions (\%)},
	y label style={font=\Huge},
    x label style={font=\Huge},    
    xticklabel style = {font=\huge},
	yticklabel style = {font=\huge},  
   symbolic x coords={$SPA$, $SPM$, $SPO$, $SBP_H$, $NNE$},
    xtick=data,
    enlarge x limits=0.2,
    bar width = 4pt,
    ytick distance = 20,
    ymax=100,
    ymin=0,
    ymajorgrids=true,
	grid style=dashed,
]	
\addplot[color=blue,fill=blue]coordinates{($SPA$,70)($SPM$,92)($SPO$,96)($SBP_H$,98)($NNE$,98)};
\addplot[color=red,fill=red]coordinates{($SPA$,70)($SPM$,92)($SPO$,96)($SBP_H$,98)($NNE$,98)};
\addplot[color=darkgreen,fill=darkgreen]coordinates{($SPA$,51.2)($SPM$,86.4)($SPO$,96)($SBP_H$,98)($NNE$,98)};
\addplot[color=orange,fill=orange]coordinates{($SPA$,76)($SPM$,92)($SPO$,96)($SBP_H$,98)($NNE$,98)};
\legend{{\leastcompd},{\leastcomp},{\random}, Max}
\end{axis}
\end{tikzpicture}}
&
\subfloat[]{
\begin{tikzpicture}
\begin{axis}[
	ybar,
	enlargelimits=0,
	legend style={legend columns=3, font=\Huge, anchor=north,at={(0.5,1.19)}},
	xlabel={Compatibility},
    ylabel={Diameter},
	y label style={font=\Huge},
    x label style={font=\Huge},    
    xticklabel style = {font=\huge},
	yticklabel style = {font=\huge},  
   symbolic x coords={$SPA$, $SPM$, $SPO$, $SBP_H$, $NNE$},
    xtick=data,
     ymin=0,
     ymax=6,
    enlarge x limits=0.2,
    bar width = 6pt,
    ymajorgrids=true,
	grid style=dashed,
]	

\addplot[color=blue,fill=blue]coordinates{($SPA$,1.867)($SPM$,2.579)($SPO$,2.644)($SBP_H$,2.729)($NNE$,2.688)};
\addplot[color=red,fill=red]coordinates{($SPA$,3)($SPM$,4.026)($SPO$,4.022)($SBP_H$,4.979)($NNE$,4.562)};
\addplot[color=darkgreen,fill=darkgreen]coordinates{($SPA$,3.987)($SPM$,4.458)($SPO$,4.61)($SBP_H$,4.717)($NNE$,4.608)};
\legend{{\leastcompd},{\leastcomp},{\random}}
\end{axis}
\end{tikzpicture}}
\\
\subfloat[]{
\begin{tikzpicture}
\begin{axis}[
	xlabel={Task size},
	ylabel={Solutions (\%)},
	y label style={font=\Huge},
    x label style={font=\Huge},   
    xticklabel style = {font=\huge},
	yticklabel style = {font=\huge}, 
	xmin= 0,
	ymin = 0,
	ytick distance = 20,
	ymax=100,
	legend style={legend columns=6, font=\Huge, anchor=north,at={(0.5,1.19)}},
	ymajorgrids=true,
	grid style=dashed,
]
	\addplot[color=darkgreen, mark=triangle*,line width=2pt,mark size=4pt]coordinates{(3,92)(5,70)(10,34)(20,0)};
	\addplot[color=blue, mark=*,line width=2pt,mark size=4pt]coordinates{(3,94)(5,92)(10,70)(20,36)};
	\addplot[color=red, mark=square*,line width=2pt,mark size=4pt]coordinates{(3,96)(5,96)(10,86)(20,62)};
	\addplot[color=cyan, mark=*,line width=2pt,mark size=4pt]coordinates{(3,98)(5,98)(10,96)(20,96)};
	\addplot[color=orange, mark=triangle*,line width=2pt,mark size=4pt]coordinates{(3,98)(5,98)(10,96)(20,96)};
\legend{{\nnp},{\mpp},{\opp},{\sbh},{\nne}}
\end{axis}
\end{tikzpicture}}
&
\subfloat[]{
\begin{tikzpicture}
\begin{axis}[
	xlabel={Task size},
	ylabel={Diameter},
	y label style={font=\Huge},
    x label style={font=\Huge},   
    xticklabel style = {font=\huge},
	yticklabel style = {font=\huge},  
	xmin= 0,
	ymin = 0,
	legend style={legend columns=6, font=\Huge, anchor=north,at={(0.5,1.19)}},
	ymajorgrids=true,
	grid style=dashed,
]
	\addplot[color=darkgreen, mark=triangle*,line width=2pt,mark size=4pt]coordinates{(3,1.83)(5,2.37)(10,3.24)};
	\addplot[color=blue, mark=*,line width=2pt,mark size=4pt]coordinates{(3,1.85)(5,2.63)(10,3.24)(20,4.06)};
	\addplot[color=red, mark=square*,line width=1.5pt,mark size=4pt]coordinates{(3,1.9)(5,2.65)(10,3.44)(20,4.13)};
	\addplot[color=cyan, mark=*,line width=2pt,mark size=4pt]coordinates{(3,1.96)(5,2.69)(10,3.44)(20,4.33)};
	\addplot[color=orange, mark=triangle*,line width=2pt,mark size=4pt]coordinates{(3,1.9)(5,2.65)(10,3.38)(20,4.02)};
\legend{{\nnp},{\mpp},{\opp},{\sbh},{\nne}}
\end{axis}
\end{tikzpicture}}
\end{tabular}}
\vspace*{-0.15in}
\caption{Team formation: comparison of algorithms ((a) and (b)), varying task size (c) and (d).}
\vspace*{-0.25in}
\label{fig:sol_perc}
\end{figure}

In Figures~\ref{fig:sol_perc}(c) and (d), we report results for teams of varying task sizes using  {\leastcompd}.
As expected, more skills means that more people need to co-operate to complete the task, making it harder to form a compatible team, and more likely to include a distant node. The number of solutions drops steeply for more strict compatibility relations, while it remains more or less constant for {\nne} and {\sbh}.

Finally, we compare our approach with previous work on team formation.
Since there is no previous work on team formation on signed network,  we create two unsigned {\epin} networks by (1) ignoring the sign of the edges and (2) deleting the negative edges.
We run a team formation algorithm \cite{Lappas09} on each of these two networks using the same tasks with $k$ = 5 skills as in the previous experiments. In Table \ref{table:baseline}, we report the percentage of
the returned teams that satisfy compatibility for the different compatibility relations. As shown, most of the teams returned are incompatible.


\vspace*{-0.1in}
\section{Related Work}
\vspace*{-0.05in}
To the best of our knowledge, our work is the first to address team formation in signed networks. 

\noindent{\textit{Team Formation.}}
%
Lappas et al. \cite{Lappas09} were the first to formally define the problem of finding a team of experts  using the network structure to quantify the quality of the team as a whole.
%
There is considerable amount of work extending their model,  (e.g., \cite{Kargar11,5590958,Anagnostopoulos:2012:OTF:2187836.2187950}),
but none of these works considers a signed network.

\noindent{\textit{Signed Networks. }}
There is a fair amount of work on signed networks ~\cite{signed-survey}.
A problem somehow related to our work is that of link and sign prediction ~\cite{social-status,Chiang2011}.
However,  we differentiate, since we are not interested in predicting future links, but rather in evaluating the compatibility between any two individuals in the network.
%
There is also work on detecting communities in signed networks (e.g., see~\cite{Yang:2007:CMS:1313049.1313219}). The notion of the team is somehow related to that of the community, but the objective of team formation is different.

\noindent{\textit{Structural Balance.}}
There is a rich literature in psychology on positive and negative relations among groups of people using structural balance theory, e.g.,~\cite{cartwright1956structural,davis1963structural}.
Structural balance has been used e.g., for identifying clusters~\cite{DBLP:conf/otm/DrummondFFL13} and polarization in networks \cite{DBLP:journals/tcss/LeeCL16}, finding communities \cite{DBLP:journals/socnet/DengAEWT16}.

\section{Conclusions}
\vspace*{-0.05in}
In this paper, we introduced the novel problem of team formation in a signed network.
The problem poses the challenge of defining node compatibility in a signed network.
To this end, we provided a principled framework by utilizing the theory of structural balance.
%
In the future, we plan to investigate different ways to combine compatibility and communication cost and to exploit compatibility for other tasks, such as link prediction or clustering.
\vspace*{-0.07in}
\section*{Acknowledgements}
The research work was supported by the Hellenic Foundation for Research and Innovation (H.F.R.I.) under the “First Call for H.F.R.I. Research Projects to support Faculty members and Researchers and the procurement of high-cost research equipment grant” (Project Number: 1873).
\bibliographystyle{ACM-Reference-Format}
\bibliography{negative-teams-short}

\end{document}